\documentclass[11pt,a4paper]{amsart}
\usepackage{mathrsfs,amssymb}
\usepackage{graphicx}

\begin{document}

\newtheorem{theorem}{Theorem}
\newtheorem{proposition}[theorem]{Proposition}
\newtheorem{lemma}[theorem]{Lemma}
\newtheorem{corollary}[theorem]{Corollary}
\newtheorem{conjecture}[theorem]{Conjecture}
\newtheorem{question}[theorem]{Question}
\newtheorem{problem}[theorem]{Problem}
\theoremstyle{definition}
\newtheorem{definition}{Definition}

\theoremstyle{remark}
\newtheorem{remark}{Remark}

\renewcommand{\labelenumi}{(\roman{enumi})}
\def\theenumi{\roman{enumi}}


\renewcommand{\Re}{\operatorname{Re}}
\renewcommand{\Im}{\operatorname{Im}}

\def \R{{\mathbb R}}
\def \H {{\mathbb H}}
\def \N {{\mathbb N}}
\def \CC {{\mathbb C}}
\def \Z {{\mathbb Z}}
\def \Q {{\mathbb Q}}
\def \TT {{\mathbb T}}

\def \L {{\mathrm L}}
\def \C {{\mathrm C}}
\def \T {{\mathrm T}}

\def \vol {\hbox{vol}}

\def\pdo{\Psi^0(D^\circ)}

\def\e{{\mathrm{e}}}

\def\i{{\mathrm{i}}}

\newcommand{\area}{\operatorname{area}}

\newcommand{\Op}{\operatorname{Op}}
\newcommand{\dom}{\operatorname{Dom}}
\newcommand{\Dom}{\operatorname{Dom}}

\newcommand{\Norm}{\mathcal N}
\newcommand{\simgeq}{\gtrsim}%
\newcommand{\simleq}{\lesssim}

\newcommand{\UN}{U_N}
\newcommand{\OPN}{\operatorname{Op}_N}
\newcommand{\HN}{\mathcal H_N}
\newcommand{\TN}{T_N}  
\newcommand{\PDO}{\Psi\mbox{DO}}

\title[Eigenfunctions of rational polygons]{Almost all eigenfunctions of a rational polygon are uniformly distributed}

\author{Jens Marklof \and Ze\'ev Rudnick}

\address{School of Mathematics, University of Bristol, Bristol BS8 1TW, UK}
\email{j.marklof@bristol.ac.uk}

\address{Raymond and Beverly Sackler School of Mathematical Sciences,
Tel Aviv University, Tel Aviv 69978, Israel}
\email{rudnick@post.tau.ac.il}

\thanks{J.M.\ was supported by a Royal Society Wolfson Research Merit Award and a Leverhulme Trust Research Fellowship. Z.R.\ was partially supported by the Israel Science Foundation (grant No. 1083/10).}

\date{30 November 2011; to appear in Journal of Spectral Theory}


\begin{abstract}
We consider an orthonormal basis of eigenfunctions of the Dirichlet Laplacian for a rational polygon. The modulus squared of the eigenfunctions defines a sequence of probability measures. We prove that this sequence contains a density-one subsequence that converges to Lebesgue measure.
\end{abstract}

\maketitle

\section{Introduction}

One of the key challenges in quantum chaos is to understand how quantum eigenstates distribute in the semiclassical limit. In the present note we will discuss this problem for eigenfunctions of the Dirichlet Laplacian for rational polygons in the plane, with particular attention to the work of Bogomolny and Schmit on ``superscars'' \cite{Bogomolny Schmit}. The conclusion of the study presented here is that any such scarring may, for almost all eigenstates, {\em only} appear in momentum, but not in configuration space.

Let $D\subset\R^2$ be a bounded connected domain with piecewise smooth boundary. The Dirichlet Laplacian $\Delta_D$ is defined as the standard Laplacian acting on functions in $\C^2(D)$ that vanish at the boundary of $D$. The eigenvalues of the positive definite operator $-\Delta_D$ (which plays the role of the quantum Hamiltonian) will be denoted by $0<E_1<E_2\leq E_3\leq\cdots\to\infty$, and the corresponding eigenfunctions by $\psi_1$, $\psi_2$, $\psi_3$, etc. A striking result on the asymptotic distribution of eigenfunctions is the Schnirelman-Zelditch-Colin de Verdi\`ere quantum ergodicity theorem, which in the present setting is due to Zelditch and Zworski \cite{ZZ}: {\em If the billiard flow is ergodic in $S^* D$ (the unit cotangent bundle of $D$), then almost all eigenfunctions become uniformly distributed in $S^* D$.}

Hassell \cite{Hassell} has shown that for certain domains $D$ with ergodic flows there are subsequences of eigenfunctions which fail to become uniformly distributed. Thus the restriction to subsequences in the above quantum ergodicity theorem is in general necessary. We are not aware of examples of domains $D$ for which the full sequence of eigenfunctions becomes uniformly distributed.

In the present note we point out that uniform distribution of the eigenfunctions in configuration space (rather than the full phase space) holds for domains given by rational polygons.  A rational polygon is a simple plane polygon (which means that its interior is connected and simply connected), so that all the vertex angles are rational multiples of $\pi$. Billiards in rational polygons give rise to dynamics that is neither integrable nor ergodic on $S^* D$ (except of course for the few integrable cases---rectangles, equilateral triangle, and right triangles with an acute vertex angle of either $\pi/3$ or $\pi/4$). The absence of ergodicity is due to the fact that, for any initial direction, the motion in $S^* D$ is   to a higher genus flat surface. By exploiting the ergodic properties of directional flows on such surfaces, we can however show that quantum ergodicity still holds in configuration space:

\begin{theorem}\label{main}
Assume $D$ is a rational polygon. Let $\{\psi_n\}$ an orthonormal basis of eigenfunctions of the Dirichlet Laplacian on $D$. Then there is a density-one sequence $n_j \in\N$ such that for any subset $A\subset D$ with boundary of measure zero,
\begin{equation}\label{maineq}
\lim_{j\to\infty} \int_A |\psi_{n_j}(x)|^2 \,dx \\
= \frac{\area(A)}{\area(D)}.
\end{equation}
\end{theorem}

Density-one means here that
\begin{equation}
\lim_{N\to\infty} \frac{\#\{ j : n_j\leq N \}}{N} =1 .
\end{equation}

The proof of Theorem \ref{main} follows Zelditch and Zworski's approach in \cite{ZZ}
until the last step, which is the only place where ergodicity of the billiard flow on $S^*D$ is used. Up to that point all arguments hold for any pseudo-differential operator of order zero. We then specialize to multiplication operators and appeal to the theorem of Kerckhoff, Masur and Smillie \cite{KMS} who showed that for rational polygons, almost all directional flows (see \S\ref{sec:polygons}) are uniquely ergodic. Details of the proof of Theorem \ref{main} are provided in \S\ref{secProof}.

The study of polygonal billiards in the context of quantum chaos goes back to Richens and Berry in \cite{Berry Richens}, who understood that the ``pseudo-integrable'' nature of the classical dynamics has an important effect on the quantum spectrum. The energy level statistics appear to be intermediate between those believed valid for generic chaotic and integrable systems, in that the level spacing distribution shows level repulsion at small distances and Poisson tails at large distances \cite{Berry Richens, Shimizu Shudo 93, Bogomolny Gerland Schmit 1999}. The quantum wave functions were investigated empirically, both numerically \cite{BU, Shimizu Shudo 95} and experimentally
using microwave cavities \cite{KS, DFMRSS}, finding ``strong
scarring" related to families of periodic orbits. Bogomolny and Schmit \cite{Bogomolny Schmit}
constructed long-lived states (quasimodes) associated  to families of
periodic orbits, which they called ``superscars", and suggested that
a positive proportion of true eigenfunctions have large overlaps with such states at
high energies. Our Theorem~\ref{main} shows that this phenomenon cannot occur in configuration space for subsequences of density larger than zero.

A recent rigorous result that holds for {\em all} eigenfunctions is a bound by Hassell, Hillairet and Marzuola \cite{HHM} which establishes that eigenfunctions cannot localize their mass away from the polygon's vertices.

\subsection*{Remarks}

\begin{enumerate}

\item Theorem \ref{main} seems new even for the classical
integrable case of the square billiard.
In that case a proof based on harmonic analysis and simple
arithmetic considerations is also available. See \cite{Jakobson} for
information on the possible quantum limits of exceptional
subsequence of eigenfunctions.

\item We can relax the requirement in our theorem that the polygon $D$ is simple, i.e.,
that its interior is simply connected. (We still require that $D$ is connected.) Then the correct definition of a rational polygon is that the group
generated by the linear parts of the reflections in the sides of the
polygon is finite. This implies that all vertex angles are rational
multiples of $\pi$, and is equivalent to the definition in the simply-connected
case, but not necessarily in the multiply connected case.

\item Theorem \ref{main} also holds for the Neumann Laplacian of $D$, and more generally for the Laplacian of an arbitrary translation surface.

\item A result of the same nature was recently obtained in \cite{RU}, by
completely different methods, for a point scatterer on the torus. For this system, the classical dynamics is integrable, but the quantum problem is not.

\end{enumerate}

\section{Background on billiards in rational polygons}
\label{sec:polygons}
For billiards in a rational polygon $D$,  the
phase space is the unit cotangent bundle $S^*D$, which is
a direct product
\begin{equation}
  S^*D = D\times S^1 = \{(x,\omega=\e^{2\pi\i\phi}): x\in D, \phi \in
  \R/\Z \} .
\end{equation}
The normalized Liouville measure is defined as
\begin{equation}
d\mu(x,\omega) = \frac 1{\area(D)}\,dx\,d\phi .
\end{equation}
The billiard flow $\Phi^t$ is defined on $S^*D$ via specular reflection for all trajectories not hitting the vertices of the polygon $D$; the reflection law for trajectories that hit a vertex can be defined arbitrarily. (The latter trajectories form a set of measure zero and are ignored in the following discussion.) The measure $d\mu$ is invariant under the billiard flow.

Let $\Gamma$ be the group generated by the linear parts of the
reflections in the sides of the polygon $D$, which   is a {\em
finite} group since our polygon is simple and all vertex angles are
rational multiples of $\pi$. For each direction $\theta$, the
set
\begin{equation}
D_\theta:=D\times \bigcup_{\gamma\in \Gamma}\{\gamma \theta\}
\end{equation}
is preserved by the flow; call this restriction $\Phi_\theta^t$ the
{\em directional flow}. Kerckhoff, Masur and Smillie showed
\cite{KMS} that for almost all $\theta$, the directional flow
$\Phi_\theta^t$ is uniquely ergodic.

Observables are (smooth) functions $a(x,\omega)$ on $S^*D$.
``Isotropic observables'' are observables which depend only on the
position variable $x$, that is are independent of the momentum
(direction):
\begin{equation}
  a(x,\omega) = a_0(x) .
\end{equation}
The time average of an observable $a\in \C^\infty(S^*D)$  is
\begin{equation}
  a^T(x,\omega) := \frac{1}{2T}\int_{-T}^T a\circ\Phi^t(x,\omega)dt.
\end{equation}

\begin{lemma}\label{dominated convergence}
Let $a(x,\omega)=a_0(x)$ be an isotropic observable. Then the time
averages satisfy
\begin{equation}\label{eq:dominated convergence}
  \lim_{T\to \infty} \int_{S^*D} |a_T|^2 d\mu
     =  \left| \frac 1{\area(D)}\int_{D} a_0(x') dx' \right|^2 .
\end{equation}
\end{lemma}
\begin{proof}

As a consequence of the Kerckhoff-Masur-Smillie theorem, there is a set of full Lebesgue measure $\Omega\subset S^1$, such that for all $(x,\omega)\in D\times\Omega$, the time
averages $a^T$ converge to the average of $a$ on $D_\omega$:
\begin{equation}
\lim_{T\to \infty}   a^T(x,\omega) =\int_{D_\omega} a d\mu_\omega=
\frac 1{\#\Gamma}\sum_{\gamma\in \Gamma} \frac 1{\area(D)}\int_{D}
 a(x',\gamma \omega)dx' .
\end{equation}
By the dominated convergence theorem we find that
\begin{equation}\label{using Kronecker}
   \lim_{T\to \infty} \int_{S^*D} |a_T|^2 d\mu
=  \int_{S^1} \left| \frac 1{\#\Gamma}\sum_{\gamma\in \Gamma} \frac
1{\area(D)} \int_{D} a(x',\gamma \omega) dx' \right|^2 d\omega .
\end{equation}
We conclude by noting that for isotropic observables, \eqref{using
Kronecker} reduces to \eqref{eq:dominated convergence}.
\end{proof}

\section{Proof of Theorem~\ref{main}\label{secProof}}

We assume throughout this section that $D$ is a rational polygon, and $\{\psi_n\}$ is a fixed orthonormal basis of eigenfunctions of the Dirichlet Laplacian on $D$.

By a standard density argument, Theorem \ref{main} is a direct consequence of the following variant for smooth isotropic observables $a(x,\omega) = a_0(x)$ with compact support in $D^\circ$ (the interior of $D$):
\begin{theorem}\label{main2}
There is a density-one sequence $n_j \in\N$ such that for any $a_0\in\C_c^\infty(D^\circ)$,
\begin{equation}\label{maineq2}
\lim_{j\to\infty} \int_D a_0(x) |\psi_{n_j}(x)|^2 \,dx \\
= \frac{1}{\area(D)}\int_{D} a_0(x)\, dx.
\end{equation}
\end{theorem}

Let $A$ be a pseudo-differential operator of order zero with
principal symbol $a$. The quantum variance for $A$  relative to the
basis $\{\psi_n\}$ is defined as
\begin{equation}
  V(A,E):=\frac 1{\#\{ E_n\leq  E\}} \sum_{ E_n\leq  E} \big|
  \langle A\psi_n,\psi_n\rangle  - \overline a \big|^2,
\end{equation}
with $\overline a := \int_{S^*D} a(x,\omega)\,d\mu(x,\omega)  $ being the phase-space average of the observable
$a$.
We recall the local Weyl law in this context \cite[Lemma 4]{ZZ}:
\begin{equation}\label{local Weyl law}
\lim_{E\to \infty} \frac 1{\#\{E_n\leq E\}} \sum_{E_n\leq E} \langle
A\psi_n,\psi_n \rangle =\overline a .
\end{equation}

As explained in \cite{ZZ}, to prove Theorem~\ref{main2}, it suffices
to show
\begin{theorem}\label{main thm}
Let $A$ be the multiplication operator $Af(x)=a_0(x) f(x)$ with $a_0\in
\C^\infty_c(D^\circ)$. Then
\begin{equation}
\lim_{E\to \infty}    V(A,E)=0 .
\end{equation}
\end{theorem}

\begin{proof}
We follow the proof of quantum ergodicity for billiards given in \cite{ZZ}.
(This elegant argument was first developed in \cite{Zelditch C^*}.) The key bound we require is
\begin{equation}\label{thebound}
\limsup_{E\to\infty} V(A,E) \leq \int_{S^*D} |a^T-\overline a|^2 d\mu
\end{equation}
for any $T>0$.
Eq.~\eqref{thebound} holds for any billiard with piecewise smooth boundary,
and any pseudo-differential operator $A$ whose Schwartz kernel is compactly supported in $D^\circ\times D^\circ$.
The proof of \eqref{thebound} follows directly from the estimates in \cite{ZZ}; we will summarize the main steps below.
Crucially, \eqref{thebound} does not require any assumptions on the ergodicity of the billiard flow.
We now specialize to $A$ whose principal symbol is an isotropic observable $a(x,\omega)=a_0(x)$,
and  input the (unique) ergodicity of almost all the directional flows in the form of Lemma~\ref{dominated convergence}.
Thus, for any $\epsilon>0$, there exists $T=T_\epsilon>0$ so that $\int_{S^*D} |a^T-\overline a|^2 d\mu <\epsilon$
and therefore $\limsup_{E\to\infty} V(A,E)<\epsilon$. This proves Theorem~\ref{main thm}.
\end{proof}

{\em Outline of the proof of eq.~\eqref{thebound}}: Fix $T>0$. Let
$U_{\delta}\subset S^* D$ be the set of initial conditions that,
within the time interval $[-2T,2T]$, do not hit a
$\delta$-neighbourhood of the vertices of the polygon $D$. We
decompose $A=A_{\delta}+R_{\delta}$, where $A_{\delta}$ is chosen so
that (i) its symbol has essential support in $U_{\delta}$, and (ii)
the measure of the essential support of the symbol of $R_{\delta}$
tends to zero as $\delta\to 0$. The principal symbol of $A_{\delta}$
is denoted by $a_{\delta}$. Since $\psi_n$ is an eigenfunction of
$\Delta_D$ we have
\begin{equation}\label{holder}
\begin{split}
|\langle (A_{\delta}-\overline a_{\delta}) \psi_n,\psi_n\rangle|^2
& = |\langle \langle A_{\delta}-\overline a_{\delta}\rangle_T \psi_n,\psi_n\rangle|^2 \\
& \leq
\langle \langle A_{\delta}-\overline a_{\delta}\rangle_T\langle A_{\delta}-\overline a_{\delta}\rangle_T^*\psi_n,\psi_n\rangle ,
\end{split}
\end{equation}
where $\langle A_\delta \rangle_T := \frac{1}{2T} \int_{-T}^T
\exp(\i t \sqrt{-\Delta_D}) A_\delta  \exp(-\i t \sqrt{-\Delta_D})
dt$. Egorov's theorem  \cite[Lemma 5]{ZZ} implies that the principal
symbol of $\langle A_\delta \rangle_T$ is  $a_{\delta}^T$. This, in
conjunction with the local Weyl law \eqref{local Weyl law} 
applied to
$\langle A_{\delta}-\overline a_{\delta}\rangle_T\langle
A_{\delta}-\overline a_{\delta}\rangle_T^*$, proves
eq.~\eqref{thebound} for the truncated $A_{\delta}$. The result is
extended to the original $A$ by showing that both sides of the
inequality are obtained as the $\delta\to 0$ limit of the truncated
version. The central ingredient in the required estimates is again
the local Weyl law; see \cite{ZZ} for details.

\end{document}